\documentclass[english,11pt, letter]{article}
\usepackage[latin9]{inputenc}
\usepackage{geometry}
\usepackage{calc}
\usepackage{amsmath}
\usepackage{amsthm}

\makeatletter

\providecommand{\tabularnewline}{\\}

\theoremstyle{plain}
\newtheorem{thm}{\protect\theoremname}
  \theoremstyle{plain}
  \newtheorem{lem}[thm]{\protect\lemmaname}

\usepackage{amsmath, amsthm, amssymb, amsfonts}
\usepackage[linesnumbered, ruled]{algorithm2e}

\makeatother

\usepackage{babel}
  \providecommand{\lemmaname}{Lemma}
\providecommand{\theoremname}{Theorem}

\begin{document}

\title{Tight Algorithms for Vertex Cover with Hard Capacities on Multigraphs
and Hypergraphs}

\author{Sam Chiu-wai Wong\\
UC Berkeley\\
samcwong@berkeley.edu}
\maketitle
\begin{abstract}
In this paper we give a $f$-approximation algorithm for the minimum
unweighted Vertex Cover problem with Hard Capacity constraints (VCHC)
on $f$-hypergraphs. This problem generalizes standard vertex cover
for which the best known approximation ratio is also $f$ and cannot
be improved assuming the unique game conjecture. Our result is therefore
essentially the best possible. This improves over the previous 2.155
(for $f=2$) and $2f$ approximation algorithms by Cheung, Goemans
and Wong (CGW).

At the heart of our approach is to apply iterative rounding to a natural
LP relaxation that is slightly different from prior works which used
(non-iterative) rounding. Our algorithm is significantly simpler and
offers an intuitive explanation why $f$-approximation can be achieved
for VCHC. We also present faster implementations of our method based
on iteratively rounding the solution to certain CGW-style covering
LPs.

We note that independent of this work, Kao \cite{kao2017iterative}
also recently obtained the same result.

\newpage
\end{abstract}

\section{Introduction}

The minimum vertex cover problem is one of the earliest NP-hard problems
studied in combinatorial optimization. In its most basic form, given
a graph $G=(V,E)$ we are asked to find a subset $U\subseteq V$,
called vertex cover, so that every edge $e\in E$ intersects $U$
in at least one of its two endpoints. The objective is to minimize
the size of $U$. Curiously, despite decades of efforts the best known
algorithms for this problem are the 2-approximation which can either
be done by LP relaxation or a simple greedy procedure. The minimum
vertex cover problem lends itself to a natural generalization to $f$-hyergraphs
where an edge $e\in E$ can have as many as $f$ endpoints. It is
not a difficult matter to generalize the 2-approximation to $f$-approximation
for this version. The seminal result of Khot showed that these algorithms
are in fact optimal assuming the Unique Game Conjecture (UGC) \cite{KhotR08}.

Chuzhoy and Naor \cite{ChuzhoyN06} initiated the study of vertex
cover with \textit{hard capacity} constraints (VCHC) where we have
a capacity of $k_{v}\geq0$ for each $v\in V$ and (a copy of) $v$
can cover at most $k_{v}$ of its incident edges. The objective is
still to minimize the size of the vertex cover found. They gave a
natural LP relaxation for VCHC from which a 3-approximation is derived
via randomized rounding for graphs with \textit{no multiple edges}.
Their analysis is based on Chebyshev inequality. Subsequently Gandhi
et al. \cite{GandhiHKKS06} improved this to a tight 2-approximation
by using Chernoff in place of Chebyshev with a much more involved
analysis. Both of these algorithms fail to work for multigraphs (graphs
possibly with multiple edges) or hypergraphs essentially because in
such cases the random variables in their analyses become unbounded
and standard concentration inequalities do not apply.

Progress had been stagnated until Saha and Khuller gave a $\min\{6f,65\}$
approximation for VCHC on hypergraphs \cite{SahaK12}. Their idea
is to apply randomized rounding for random variables at different
scales to salvage Chernoff. Partly inspired by their result, Cheung,
Goemans and Wong (CGW) surprisingly gave simple deterministic rounding
algorithms which achieve significantly better approximation ratios
of 2.155 (for graphs) and $2f$ \cite{cheung2014improved}. Their
method is to formulate the coverage requirement of randomized rounding,
used in all previous works, in terms of another LP and study the property
of its extreme point solutions. In other words, their approach is
a 2-stage LP rounding procedure which solves the same LP relaxation
followed by the ``coverage requirement LP''.

\subsection{Our contribution}

We propose a new simple approach to the problem based on iterative
rounding without using new LPs. Our algorithm achieves the best possible
approximation ratio $f$ and essentially settles its approximability.
Our approach is inspired by ideas used in previous works, most notably
CGW which considers an extreme point solution to certain covering
LPs. In hindsight their method suggested the possibility of a better
approximation obtained by iterative rounding, which often exploits
the structure of extreme point solutions. We also show that when combined
with iterative rounding, CGW approach can be extended to give another
$f$-approximation. Although more contrived, this alternate algorithm
may be preferred as it involves iteratively rounding the solution
to so-called \textit{covering LP}s, which can be solved faster than
general LPs using dedicated algorithms \cite{plotkin1995fast}.

\begin{center}
\begin{tabular}{|c|c|c|}
\hline 
Authors & $\begin{array}{c}
\text{approx. ratio}\\
\text{(graphs, hypergraphs)}
\end{array}$ & multigraphs and hypergraphs okay?\tabularnewline
\hline 
\hline 
Chuzhoy, Naor \cite{ChuzhoyN06} & 3, {*} & no\tabularnewline
\hline 
Gandhi et al. \cite{GandhiHKKS06} & 2, {*} & no\tabularnewline
\hline 
Saha, Khuller \cite{SahaK12} & 12, $\min\{6f,65\}$  & yes\tabularnewline
\hline 
Cheung, Goemans, Wong \cite{cheung2014improved} & 2.155, $2f$ & yes\tabularnewline
\hline 
This paper, \cite{kao2017iterative} & 2, $f$ & yes\tabularnewline
\hline 
\end{tabular}
\par\end{center}

\subsection{Other related works}

Prior to the work of Chuzhoy and Naor \cite{ChuzhoyN06} which initiated
the study of vertex cover with \textit{hard} capacity constraints,
Guha et al. \cite{GuhaHKO02} resolved the problem with \textit{soft}
capacity constraints (where a vertex can be used an arbitrary number
of times) using a clever primal-dual algorithm. Notably, their result
holds even for the weighted setting whereas the hard capacity version
is as hard as set cover in the weighted case and an approximation
ratio of $O(\log n)$ is optimal \cite{ChuzhoyN06}. Using dependent
randomized rounding, another 2-approximation for the soft capacity
version was given by Gandhi et al. \cite{GandhiKPS02}.

\section{Preliminaries}

Let $G=(V,E)$ be a multigraph. We write $u\in e$ to indicate that
$u$ is an endpoint of edge $e\in E$.

The minimum Vertex Cover problem with Hard Capacity constraints (VCHC)
is specified by $(V,E,k,m)$, where 
\begin{itemize}
\item $G=(V,E)$ is the input multigraph, 
\item For each $v\in V$, $m_{v}$ denotes the maximum number of copies
of $v$ one can select, 
\item For each $v\in V$, $k_{v}$ is the number of incident edges (a copy
of) $v$ can cover. 
\end{itemize}
A solution to VCHC consists of $(x,y)=(\{x_{v}\}_{v\in V},\{y(e,v)\}_{e\in E,v\in e})$.
Here $x_{v}$ is the number of copies of vertex $v$ selected, and
the assignment variable $y(e,v)\in\{0,1\}$ represents whether edge
$e$ is covered by $v$, for each $e\in E$ and $v\in e$. A solution
$(x,y)$ is \emph{feasible} for VCHC if
\begin{enumerate}
\item For all $v\in V$: $x_{v}\in\{0,1,\cdots,m_{v}\}$, 
\item For all $e\in E$: $\sum_{v\in e}y(e,v)=1$ (i.e.~any edge must be
covered by one of its endpoints), 
\item For all $v\in V$: $|\{e:y(e,v)=1\}|\leq k_{v}x_{v}$ (i.e.~the total
number of edges assigned to $v$ does not exceed its total capacity). 
\end{enumerate}
The objective of VCHC is to find a feasible solution $(x,y)$ for
VCHC that minimizes $\sum_{v\in V}x_{v}$, the size of the vertex
cover. As VCHC generalizes the classical minimum vertex cover problem
which is already NP-hard, we provide efficient algorithms for finding
good approximate solutions. Our approach is based on rounding a fractional
solution to the following \textbf{LP1} relaxation, which has been
used extensively in the literature \cite{ChuzhoyN06,GandhiHKKS06,SahaK12,cheung2014improved}.
\begin{subequations}\label{LP:VCHC} 
\begin{alignat}{2}
\text{min } & \sum_{v\in V}x_{v}\nonumber \\
\text{s.t. } & \sum_{v\in e}y(e,v)=1 & \quad & \forall e\in E\label{eq:1.1}\\
 & y(e,v)\leq x_{v} & \quad & \forall e\in E,v\in e\label{eq:1.2}\\
 & \sum_{e\in\delta(v)}y(e,v)\leq k_{v}x_{v} & \quad & \forall v\in V\label{eq:1.3}\\
 & x_{v}\leq m_{v} & \quad & \forall v\in V\label{eq:1.4}\\
 & x,y\geq0\label{eq:1.5}
\end{alignat}
\end{subequations}Here $x_{u}$ denotes the number of copies of $u$
selected and $y(e,v)$ indicates whether $e$ is covered by $v$.
The first constraint says that each $e\in E$ should be covered by
one of its endpoints, the second ensures that $e$ can be covered
by a vertex selected, and the third is the capacity constraint.

The following lemma shows that, when constructing a feasible solution
to VCHC, we only need the integrality of $x$, and not of $y$. This
follows easily by the integrality of flows in networks with integer
capacities. We refer readers to \cite{ChuzhoyN06,SahaK12} for a proof.
\begin{lem}[Chuzhoy and Naor \cite{ChuzhoyN06}, generalized to hypergraphs by
Saha and Khuller \cite{SahaK12}]
\label{b-matching} If $(x,y)$ is feasible for LP1, and $x$ is
integral, there exists an integral $y'$ such that $(x,y')$ is feasible
for LP \ref{LP:VCHC}, and $y'$ can be found efficiently by a maximum
flow computation.
\end{lem}
In light of this lemma, it suffices to identify a feasible integral
solution $x$ with a good approximation guarantee.

\section{$f$-approximation for VCHC on $f$-hypergraphs}

Let $(x^{*},y^{*})$ be an optimal extreme point solution to LP1,
and $U=\{u\in V:x_{u}^{*}\geq1/f\}$. A natural idea used in all previous
works is round up $u\in U$ which involves only a factor $f$ blowup,
and select judiciously a subset of $W=\{w\in V:0<x_{w}^{*}<1/f\}$.

\paragraph{Covering tight edges}

Our iterative rounding scheme\footnote{See e.g. \cite{lau2011iterative} for the background on iterative
rounding which was introduced by Jain \cite{jain2001factor}.} is based on the observation that a \textit{tight} edge $e\in\delta(u)$
with $y(e,u)=x_{u}\geq1/f$ can be rounded up while respecting the
capacity constraint. This follows from the capacity constraint used
in the LP1 where R.H.S. is $k_{u}x_{u}$. Therefore we may effectively
remove $e$ from the LP by covering $e$ with $u$ and decreasing
$k_{u}$ by 1, and solve the new smaller LP relaxation. A similar
argument was used in the (non-iterative) rounding algorithms in \cite{ChuzhoyN06,GandhiHKKS06}.

Nevertheless, one complication arises as any $x_{u}^{*}\geq1/f$ can
in principle drop below $1/f$ in later iterations of the algorithm
and end up not being selected, i.e. $x_{u}=0$ in the final solution.
In this case covering $e$ by $u$ is not justified. Here we introduce
the constraint $1/f\leq x_{u}$ to the rescue. It ensures that any
$u$ with $x_{u}^{*}\geq1/f$ will stay above $1/f$ ever after.

\paragraph{Fixing $x_{u}^{*}=1/f$ }

To further simplify the LP, we observe that any $x_{u}^{*}=1/f$ can
be readily rounded up and removed from the LP. In terms of cost this
is a good idea as the approximation ratio incurred is exactly $f$,
meaning that we are not being lossy. Moreover, it ensures that any
edge would always have an endpoint in $U$ (see Lemma \ref{lem:intersectU})
which is important when we bound the approximation ratio by exploiting
the structure of the extreme point solution in the proof of Lemma
\ref{lem:extpt}. The idea of examining an extreme point solution
was inspired by \cite{cheung2014improved}.

\paragraph{Modified LP relaxation}

We incorporate these insights into \textbf{LP2} below. LP2 resembles
the form of LP1 with a few important modifications. The first constraint
involves $\bar{y}(e,v)$ which is the coverage of $v$ towards $e$
in the final solution. This is a result of fixing $x_{u}^{*}=1/f$
as discussed above. The third constraint has $k_{u}-|T_{u}|$ in place
of $k_{u}$; here $|T_{u}|$ is the number of tight edges covered
by $u$ so we simply subtract $|T_{u}|$ from the capacity $k_{u}$.
The fourth constraint now sums over only non-tight edges as any tight
edges have already benn covered by its endpoint in $U$. Finally a
new constraint $x_{u}\geq1/f$ is introduced to ensure that $x_{u}^{*}\geq1/f$
cannot drop below $1/f$. The last new constraint $x_{w}\leq1/f$
is, strictly speaking, not needed but is included to simplify our
exposition.

\begin{subequations}\label{LP:VCHC-2}

\begin{alignat}{2}
\text{min } & \sum_{v\in V\backslash D}x_{v}\nonumber \\
\text{s.t. } & \sum_{v\in e\backslash D}y(e,v)=1-\sum_{v\in e\cap D}\bar{y}(e,v) & \quad & \forall e\in E\backslash T\label{eq:1'.1}\\
 & y(e,v)\leq x_{v} & \quad & \forall e\in E\backslash T,v\in e\backslash D\label{eq:1.2-1}\\
 & \sum_{e\in\delta(u)\backslash T}y(e,u)\leq(k_{u}-|T_{u}|)x_{u} & \quad & \forall u\in U_{>}\label{eq:1.3-1}\\
 & \sum_{e\in\delta(w)\backslash T}y(e,w)\leq k_{w}x_{w} & \quad & \forall w\in W\\
 & 1/f\leq x_{u}\leq m_{u} & \quad & \forall u\in U_{>}\label{eq:1.4-1}\\
 & 0\leq x_{w}\leq1/f & \quad & \forall u\in W\label{eq:1.5-1}\\
 & y\geq0
\end{alignat}
\end{subequations}

\paragraph{Algorithm}

We adopt the following notations:
\begin{itemize}
\item $U=\{u\in V:x_{u}^{*}\geq1/f\}$, $W=\{w\in V:0<x_{w}^{*}<1/f\}$,
$Z=\{z\in V:x_{z}^{*}=0\}$. 
\item Further divide $U$ into $U_{>}=\{u\in U:x_{u}^{*}>1/f\}$ and $U_{=}=\{u\in U:x_{u}^{*}=1/f\}$.
\item $T=\bigcup_{u\in V}T_{u}$ is a disjoint union of edges $e\in T_{u}$
covered by $u\in V$. We call $T$ the set of tight edges.
\item $D\subseteq V$ is the set of vertex $v$ whose $\bar{x}_{v}$ has
been determined.
\end{itemize}
Our algorithm is based on performing iterative rounding on LP2. It
incrementally builds up a feasible solution $(\bar{x},\bar{y})$ to
VCHC for which $\bar{x}$ is integral.

\begin{algorithm}[H]
\SetAlgoLined   
Solve LP1 for an extreme point solution $(x^{*},y^{*})$. Initially $T_{u}=\emptyset$ and $D=\emptyset$; initialize $U,U_{>},U_{=},W,Z$ based on $(x^{*},y^{*})$.\;
\Repeat{$U_{=}=Z=\emptyset$ and $y^{*}(e,u)<x_{u}^{*}\forall u\in U,e\in\delta(u)\backslash T$}{
\For{$v\in Z$}{set $\bar{x}_{v}=x_{v}^{*}=0$, $\bar{y}(e,v)=y^{*}(e,v)=0$ for $e\in\delta(v)$\; $D\longleftarrow D\cup\{v\}$\;}
\If{$y^{*}(e,u)=x_{u}^{*}$ for some $u\in U$ and $e\in\delta(u)\backslash T$} {set $\bar{y}(e,u)=1$, $\bar{y}(e,v)=0$ for $v\in e\backslash\{u\}$\; $T_{u}\longleftarrow T_{u}\cup\{e\}$\;}
\For{$u\in U_{=}$} {set $\bar{x}_{u}=1$, $\bar{y}(e,u)=y^{*}(e,u)$ for $e\in\delta(u)\backslash T$\; $D\longleftarrow D\cup\{u\}$\;}
Solve updated LP2 for a new extreme point solution $(x^{*},y^{*})$. Update $U,U_{>},U_{=},W,Z$ based on $(x^{*},y^{*})$.
}
For $v\notin D$, set $\bar{x}_{v}=\lceil x_{v}^{*}\rceil$ and $\bar{y}(e,v)=y^{*}(e,v)$ for $e\in\delta(v)\backslash T$
\caption{Iterative Rounding Algorithm for VCHC} 
\end{algorithm}

\paragraph{Analysis}

The algorithm works mostly by design. First we demonstrate feasibility.
\begin{lem}
\label{lem:stayabove}Suppose a vertex $u$ satisfies $x_{u}^{*}\geq1/f$
at some time during the execution of the algorithm. We must then have
$x_{u}^{*}\geq1/f$ after so long as $u\notin D$. Moreover, in the
final solution $\bar{x}_{u}\geq1$.
\end{lem}
\begin{proof}
As soon as $x_{u}^{*}\geq1/f$, we either have $u\in U_{=}$ or $u\in U_{>}$.
In the former case, $u$ is immediately inserted into $D$ (line 13)
and $\bar{x}_{u}=1$. In the latter case the constraint $x_{u}\ge1/f$
ensures $x_{u}^{*}\geq1/f$ ever after. Eventually we either have
$u\in U_{=}$ (so $\bar{x}_{u}=1$) or $\bar{x}_{u}=\lceil x_{u}^{*}\rceil\geq\lceil1/f\rceil\geq1$. 
\end{proof}
\begin{lem}
\label{lem:feas}The final solution $(\bar{x},\bar{y})$ is feasible
with $\bar{x}$ integral.
\end{lem}
\begin{proof}
The fact that $\bar{x}$ is integral simply follows from the description
of the algorithm.

We first argue that all edges are covered. For tight edges $e\in T$,
we must have set $\bar{y}(e,u)=1$ for some $u\in U$ at some point
(line 8) so $e$ is covered by $u$. For non-tight edges $e\notin T$,
when the algorithm terminates we have
\[
\sum_{v\in e\backslash D}\bar{y}(e,v)=\sum_{v\in e\backslash D}y^{*}(e,v)=1-\sum_{v\in e\cap D}\bar{y}(e,v)
\]
so $e$ is indeed covered.

It remains to argue that the capacity constraint is satisfied. Lines
3-6 and 11-14 are clearly okay. For Lines 7-10, by Lemma \ref{lem:stayabove}
any $u\in U$ satisfies $\bar{x}_{u}\geq1$ so we may simply set $\bar{y}(e,u)=1$
and subtract 1 from the capacity $k_{u}$ in the LP. This is exactly
why we have $(k_{u}-|T_{u}|)x_{u}$ in the third type of constraints.
\end{proof}
Now we bound the approximation ratio. Our argument consists of two
ingredients. The first is to observe that an edge $e\notin T$ always
intersects $U$ as only vertices with $x_{v}\leq1/f$ is put into
$D$. The second, which is much more crucial, exploits the structure
of an extreme point solution to show that there cannot be too many
fractional $x_{w}^{*}$ left at the end of the while loop.
\begin{lem}
\label{lem:oldfeasfornew}In line 15 of the algorithm, the old $(x^{*},y^{*})$
(from before this line but restricted to only the variables appearing
in updated LP2) is still feasible for updated LP2.
\end{lem}
\begin{proof}
Clear by inspection.
\end{proof}
\begin{lem}
\label{lem:intersectU}When the algorithm terminates, for any $e\notin T$
we have $\sum_{v\in e\backslash D}x_{v}^{*}\geq\sum_{v\in e\backslash D}y^{*}(e,v)\ge|e\backslash D|/f$.
\end{lem}
\begin{proof}
Note that vertices are assigned to $D$ in Lines 3-6 and 11-14, where
we have $x_{v}^{*}\in\{0,1/f\}$ (note that this $x_{v}^{*}$ is the
one from that particular iteration of the algorithm and not necessarily
the final one). Therefore $\bar{y}(e,v)=y^{*}(e,v)\leq x_{v}^{*}\leq1/f$,
which implies 
\[
\sum_{v\in e\backslash D}x_{v}^{*}\geq\sum_{v\in e\backslash D}y^{*}(e,v)=1-\sum_{v\in e\cap D}\bar{y}(e,v)\ge1-|e\cap D|/f\geq|e\backslash D|/f,
\]
where the last inequality follows from $|e|\leq f$.
\end{proof}
Now we show that there cannot be too many elements in $W$ by a simple
counting argument based on examining the extreme point.
\begin{lem}
\label{lem:extpt}When the algorithm terminates, $|W|\leq|U^{=}|$
where $U^{=}:=\{u\in U:x_{u}^{*}=m_{u}\}$.
\end{lem}
\begin{proof}
As an extreme point solution, $(x^{*},y^{*})$ is obtained by setting
some of the constraints as equalities. We call these constraints,
which form an \textbf{invertible} matrix, \textbf{tight}. The proof
is based on examining the structure and number of these tight constraints.

First note that $0\leq x_{w}\leq1/f$ and $1/f\leq x_{u}$ cannot
be tight since $U_{=}=Z=\emptyset$. Similarly, $y(e,u)\leq x_{u}$
for $u\in U$ cannot be tight since no more edges can be added to
$T$. Furthermore, we disregard any edge $e\subseteq D$ since all
of its $y(e,v)$ has been determined. 

Observe that by the last lemma each edge $e\notin T$ (with $e\backslash D\neq\emptyset$)
must have an $x_{v}^{*}\ge y^{*}(e,v)\ge1/f$ . In other words, $1+|e\cap W|\leq|e\backslash D|$.

Below we count the number of tight constraints of different types.
The first, second and third in the table are self-explanatory.

For the fourth one in the table there is a total of $\sum_{e\notin T}|e\cap W|+|W|$
such constraints but we claim that only $\sum_{e\notin T}|e\cap W|$
can be tight. This follows from the fact that for each $w\in W$,
setting all of the $1+|\delta(w)\backslash T|$ corresponding constraints
tight would give a singular system. More precisely, this would give
$x_{w}=y(e,w)=0$ which renders setting the constraint $\sum_{e\in\delta(w)\backslash T}y(e,w)\leq k_{w}x_{w}$
tight unnecessary. In other words, including all of these constraints
would give rise to a singular matrix. Thus the total number is at
most $\sum_{w}|\delta(w)\backslash T|=\sum_{e\notin T}|e\cap W|$.

For the fifth one, by Lemma \ref{lem:intersectU} each edge $e\notin T$
(with $e\backslash D\neq\emptyset$) must satisfy $x_{v}^{*}\ge y^{*}(e,v)\ge1/f$
for some $v\in e$. Therefore at least one of the $|e\cap U_{>}|$
constraints $y(e,u)\geq0$ is not tight.

\begin{tabular}{|c|c|}
\hline 
constraints & \#tight ones\tabularnewline
\hline 
\hline 
$\sum_{v\in e\backslash D}y(e,v)=1-\sum_{v\in e\cap D}\bar{y}(e,v)$ & $=|E\backslash T|$\tabularnewline
\hline 
$\sum_{e\in\delta(u)\backslash T}y(e,u)\leq(k_{u}-|T_{u}|)x_{u}$ & $\leq|U_{>}|$\tabularnewline
\hline 
$x_{u}\leq m_{u}$ & $=|U^{=}|$\tabularnewline
\hline 
$y(e,w)\leq x_{w}$, $y(e,w)\geq0$ and $\sum_{e\in\delta(w)\backslash T}y(e,w)\leq k_{w}x_{w}$
($w\in W$) & $\leq\sum_{e\notin T}|e\cap W|$\tabularnewline
\hline 
$y(e,u)\geq0$ ($u\in U_{>}$)  & $\leq\sum_{e\notin T}|e\cap U_{>}|-1$\tabularnewline
\hline 
\end{tabular}

Now the number of tight constraints is at most 
\begin{eqnarray*}
|E\backslash T|+|U_{>}|+|U^{=}|+\sum_{e\notin T}\left(|e\cap W|+|e\cap U_{>}|-1\right) & = & |E\backslash T|+|U_{>}|+|U^{=}|+\sum_{e\notin T}\left(|e\backslash D|-1\right)\\
 & = & \sum_{e\notin T}|e\backslash D|+|U_{>}|+|U^{=}|.
\end{eqnarray*}

On the other hand, the number of variables is $|U_{>}|+|W|+\sum_{e\notin T}|e\backslash D|$.
Since there is an equal number of tight constraints and variables,
we have $|W|\leq|U^{=}|$.
\end{proof}
We are ready to derive our main theorem.
\begin{thm}
\label{thm:main}Our algorithm is a $f$-approximation.
\end{thm}
\begin{proof}
Feasibility follows from Lemmas \ref{lem:feas} and \ref{b-matching}.
We bound the approximation ratio. By Lemma \ref{lem:oldfeasfornew}
the old $(x_{old}^{*},y_{old}^{*})$ is still feasible for the updated
LP so the objective value of the new $(x_{new}^{*},y_{new}^{*})$
is no worse: 
\[
\text{cost}(x_{new}^{*},y_{new}^{*})=\sum_{v\in V\backslash D_{new}}x_{new,v}^{*}\leq\sum_{v\in V\backslash D_{new}}x_{old,v}^{*}=\text{cost}(x_{old}^{*},y_{old}^{*})-\sum_{v\in U_{old,=}\cup Z}x_{old,v}^{*}
\]
which says that the cost $|U_{old,=}|$ incurred to round up vertex
in $U_{old,=}$ can be charged to $\sum_{v\in U_{old,=}\cup Z}x_{old,v}^{*}=|U_{old,=}|/f$
with a factor $f$ blowup.

For the last $(x_{last}^{*},y_{last}^{*})$, the cost of rounding
up the remaining $\bar{x}_{v}=\lceil x_{last,v}^{*}\rceil$ is
\begin{eqnarray*}
|W|+\sum_{u\in U^{=}}m_{u}+\sum_{u\in U\backslash U^{=}}\lceil x_{last,u}^{*}\rceil & \leq & |U^{=}|+\sum_{u\in U^{=}}m_{u}+\sum_{u\in U\backslash U^{=}}\lceil x_{last,u}^{*}\rceil\\
 & \leq & f\left(\sum_{u\in U^{=}}m_{u}+\sum_{u\in U\backslash U^{=}}x_{last,u}^{*}\right)\\
 & = & f\cdot\text{cost}(x_{last}^{*},y_{last}^{*}),
\end{eqnarray*}
where we used Lemma \ref{lem:extpt} and $m_{u}\ge1,f\geq2,x_{last,u}^{*}\geq1/f$.
This proves the theorem.
\end{proof}

\subsection{Implementation using only the original LP}

It is possible to implement a similar algorithm without appealing
to LP2, which essentially introduces the new constraint $x_{u}\geq1/f$.
Without them previous $u\in U$ may fall into $W$ and $u$ may not
be selected in the final solution, in which case setting $\bar{y}(e,u)=1$
for $y^{*}(e,u)=x_{u}^{*}$ is not justified as $u$ does not have
the capacity to cover $e$.

The key idea is to continuously move from the old optimum to the new
one. More concretely, consider moving along the line from $(x_{old}^{*},y_{old}^{*})$
to $(x_{new}^{*},y_{new}^{*})$. We stop whenever $x_{v}^{*}=1/f$
at an intermediate point $(x^{*},y^{*})$, where we perform lines
11-14 by fixing $x_{v}=1$ and $y(e,u)=y^{*}(e,u)$, and solve the
updated LP again. If we arrive at $(x_{new}^{*},y_{new}^{*})$ without
stopping, then we perform lines 7-10 by covering any tight edge $y^{*}(e,u)=x_{u}^{*}\geq1/f$
using $u$, i.e. removing $e$ from the LP and decreasing $k_{u}$
by 1; and solve the updated LP again (one can also eliminate $Z$
which, as with the previous approach, is only introduced to simplify
the exposition). If none of these operations are possible, then we
are at an extreme point where the same proof would show that rounding
up the remaining vertices would give a $f$-approximation.

Readers can easily make the previous proof work for this new implementation.
We refrain from giving a proof because in the next section, we would
give a faster implementation using this idea on another LP with a
full proof.

\section{Faster Implementation via Solving Covering LPs}

The approach in the previous section essentially redistributes the
coverage relation $y(e,v)$ and cost $x_{v}$ from one iteration to
the next. Upon a closer examination, readers may have noticed that
the constraint $y(e,w)\leq x_{w}$ for $w\in W$ plays a relatively
smaller role. One may wonder if there could a faster implementation
without solving LP2 again. We answer this in the affirmative in this
section. This alternate approach performs iterative rounding on LP3,
which is a packing LP and can be solved faster using various specialized
algorithms. The formulation of LP3 has been heavily inspired by a
similar construction in \cite{cheung2014improved}.

\subsection{LP relaxation and algorithm}

In this section we use the same notations $U,U_{>},U_{=},W,Z$ as
before. Our covering \textbf{LP3} is as follows:

\begin{subequations}\label{LP2-2} 
\begin{alignat}{2}
\text{min } & \qquad\qquad\qquad\qquad\qquad\sum_{w\in W}x_{w}+\sum_{u\in U_{>}}x_{u}\nonumber \\
\text{s.t. } & \sum_{w\in W}M(u,w)x_{w}+(k_{u}-|T_{u}|)x_{u}\geq\sum_{w\in W}M(u,w)x_{w}^{*}+\sum_{e\in\delta(u)\backslash T}y^{*}(e,u) & \quad & \forall u\in U_{>}\label{eq:2.1-2}\\
 & \qquad\qquad\qquad\qquad\qquad0\leq x_{w}\leq1/f & \quad & \forall w\in W\label{eq:2.2-2}\\
 & \qquad\qquad\qquad\qquad\qquad1/f\leq x_{u}\leq m_{u} & \quad & \forall u\in U_{>}\label{eq:2.3-2}
\end{alignat}

\end{subequations}where
\[
M(u,w)=\sum_{e\in E_{u}\cap\delta(w)}\frac{y^{*}(e,w)}{x_{w}^{*}}.
\]

\paragraph{Redistributing coverage}

At the high level, LP3 is based on doing book-keeping of how edges
$\delta(u)$ incident to $u\in U$ are covered. It attempts to redistribute
the coverage by maintaining the proportion of capacity used for different
$y(e,w)$'s (for a given $w$). While an edge $e$ can have more than
one endpoint in $U$, we simply arbitrarily assign $e$ to one such
$u$ so that $E_{u}$ is the collection of edges assigned to $u$.

For $w\in W$ we distribute its capacity in the same proportion as
before (hence step 4(a) and the definition of $M(u,w)$; see also
Lemma \ref{lem:capFeas2}). The coverage $y(e,u)$ for $e\in E_{u}$
would then be the remaining amount not yet covered. However this amount
can become negative if the other endpoints of $e$ in $W$ contribute
more than before towards covering $e$. Similarly, if they contribute
much less now $y(e,u)$ can become larger than $x_{u}$.

\textbf{The key idea is to slowly move from an old solution $(x^{*},y^{*})$
to the new $(x_{new}^{*},y_{new}^{*})$, and stop whenever any of
these desired conditions is about to fail (step 5).} At this point
we may simplify the solution by appropriately modifying $x,y$.

The algorithm is as follows. Readers should recognize the resemblance
to the previous one, except with the notable difference that we need
to ensure $y(e,u)\geq0$ in addition to $x_{u}\geq1/f$ and $y(e,u)\leq x_{u}$
(pardon us for not using the algorithm environment here as it would
clutter the longer description).

\begin{center}
\noindent\fbox{\begin{minipage}[t]{1\columnwidth - 2\fboxsep - 2\fboxrule}%
\begin{center}
\textbf{Iterative Rounding Algorithm using Covering LPs}
\par\end{center}
\begin{enumerate}
\item Solve \textbf{LP1} for an extreme point solution $(x^{*},y^{*})$.
Initially $T_{u}=\emptyset$ and $D=\emptyset$; initialize $U,U_{>},U_{=},W,Z$
based on $(x^{*},y^{*})$.
\item For $v\in Z$, set $\bar{x}_{v}=x_{v}^{*}=0$, $\bar{y}(e,v)=y^{*}(e,v)=0$
for $e\in\delta(v)$ and $D\longleftarrow D\cup\{v\}$.
\item Partition edges into a disjoint union $E\backslash T=\bigcup_{u\in U}E_{u}$
by assigning $e\notin T$ to an arbitrary $E_{u}$ where $u\in e\cap U$.
\item Solve updated \textbf{LP3} for an extreme point solution $x_{new}^{*}$
and set 

\begin{enumerate}
\item $y_{new}^{*}(e,w)=\frac{y^{*}(e,w)}{x_{w}^{*}}x_{new,w}^{*}$ for
$w\in W,e\in\delta(w)\backslash T$ 
\item $y_{new}^{*}(e,u)=y^{*}(e,u)+\sum_{w\in e\cap W}\left(y^{*}(e,w)-y_{new}^{*}(e,w)\right)$
for $u\in U,e\in E_{u}$ 
\item $y_{new}^{*}(e,u)=y^{*}(e,u)$ for $u\in U,e\notin E_{u}\cup T$
\end{enumerate}
\item Let $x^{t}=(1-t)x^{*}+tx_{new}^{*}$ and $y^{t}=(1-t)y^{*}+ty_{new}^{*}$,
where $t$ \textbf{continuously} increases from 0 to 1. \textbf{Stop}
whenever:

\begin{enumerate}
\item $y^{t}(e,u)=x_{u}^{t}$ for $u\in U$ and $e\in\delta(u)\backslash T$.
Set $\bar{y}(e,u)=1$, $\bar{y}(e,v)=0$ for $v\in e\backslash\{u\}$
and $T_{u}\longleftarrow T_{u}\cup\{e\}$.
\item $x_{v}^{t}=1/f$ for some $v$. Set $\bar{x}_{u}=1$, $\bar{y}(e,u)=y^{*}(e,u)$
for $e\in\delta(u)\backslash T$ and $D\longleftarrow D\cup\{u\}$.
\item $y^{t}(e,u)=0$ for $u\in U$ and $e\in\delta(u)\backslash T$. Set
$\bar{y}(e,u)=0$ and $e\longleftarrow e\backslash\{u\}$.
\end{enumerate}
\item Update $U,U_{>},U_{=},W,Z$ based on $(x^{*},y^{*})\longleftarrow(x^{t},y^{t})$.
\item \textbf{Repeat} steps 2-7 until no more updates are possible (reaching
$x_{new}^{*}$ without stopping).
\item For $v\notin D$, set $\bar{x}_{v}=\lceil x_{v}^{*}\rceil$ and $\bar{y}(e,v)=y^{*}(e,v)$
for $e\in\delta(v)\backslash T$.
\end{enumerate}
\end{minipage}}
\par\end{center}

First we show that our algorithm is well-defined. Interestingly, this
is analogous to Lemma \ref{lem:intersectU} which is used to establish
approximation guarantee instead.
\begin{lem}
\label{lem:interU2}We have $\sum_{v\in e\backslash D}x_{v}^{*}\geq\sum_{v\in e\backslash D}y^{*}(e,v)\ge|e\backslash D|/f$.
In particular, there is some $u\in e\cap U$ so step 3 of the algorithm
is well-defined.
\end{lem}
\begin{proof}
Note that vertices are assigned to $D$ in steps 2 and 5(b), where
we have $x_{v}^{*}\leq1/f$. Therefore $\bar{y}(e,v)=y^{*}(e,v)\leq x_{v}^{*}\leq1/f$,
which implies 
\[
\sum_{v\in e\backslash D}x_{v}^{*}\geq\sum_{v\in e\backslash D}y^{*}(e,v)=1-\sum_{v\in e\cap D}\bar{y}(e,v)\ge1-|e\cap D|/f\geq|e\backslash D|/f,
\]
where the last inequality follows from $|e|\leq f$.
\end{proof}
Now we argue for feasibility. As before, any $x_{u}^{*}\geq1/f$ would
stay above $1/f$ which is necessary to cover tight edges in step
5(a).
\begin{lem}
\label{lem:stayabove2}Suppose $u$ satisfies $x_{u}^{*}\geq1/f$
at some time during the execution of the algorithm. We must then have
$x_{u}^{*}\geq1/f$ after so long as $u\notin D$. Moreover, in the
final solution $\bar{x}_{u}\geq1$.
\end{lem}
\begin{proof}
Any such $x_{u}^{*}$ cannot drop below $1/f$ thanks to step 5(b)
of the algorithm, and would be rounded up eventually.
\end{proof}
We justify the capacity constraint in the next lemma, which goes hand-in-hand
with the way \textbf{LP3} is formulated.
\begin{lem}
\label{lem:capFeas2}$(x^{*},y^{*}),(x_{new}^{*},y_{new}^{*}),(x^{t},y^{t})$
satisfy the capacity constraint $\sum_{e\in\delta(w)\backslash T}y(e,w)\leq k_{w}x_{w}$
for $w\in W$ and $\sum_{e\in\delta(u)\backslash T}y(e,u)\leq(k_{u}-|T_{u}|)x_{u}$
for $u\in U$.
\end{lem}
\begin{proof}
We proceed by induction.

For $w\in W$, we have $\sum_{e\in\delta(w)\backslash T}y_{new}^{*}(e,w)=\sum_{e\in\delta(u)\backslash T}\frac{y^{*}(e,w)}{x_{w}^{*}}x_{new,w}^{*}\leq k_{w}x_{new,w}^{*}$
since by the induction hypothesis we have $\sum_{e\in\delta(w)\backslash T}y_{new}^{*}(e,w)\leq k_{w}x_{w}^{*}$.
Now $(x^{t},y^{t})$ (and therefore the new $(x^{*},y^{*})$) also
satisfy the capacity constraint since it is a convex combination of
$(x^{*},y^{*}),(x_{new}^{*},y_{new}^{*})$.

For $u\in U$, we have 
\begin{eqnarray*}
\sum_{e\in\delta(u)\backslash T}y_{new}^{*}(e,u) & = & \sum_{e\in E_{u}}\left[y^{*}(e,u)+\sum_{w\in e\cap W}\left(y^{*}(e,w)-y_{new}^{*}(e,w)\right)\right]+\sum_{e\notin E_{u}\cup T}y^{*}(e,u)\\
 & = & \sum_{e\in\delta(u)\backslash T}y^{*}(e,u)+\sum_{w\in e\cap W}\left(y^{*}(e,w)-\frac{y^{*}(e,w)}{x_{w}^{*}}x_{new,w}^{*}\right)\\
 & \leq & (k_{u}-|T_{u}|)x_{new,u}^{*}
\end{eqnarray*}
where the last inequality follows from the first constraint of \textbf{LP3}.
Hence $(x_{new}^{*},y_{new}^{*})$ satisfies the capacity constraint
for $u\in U$, and so is $(x^{t},y^{t})$ since it is a convex combination
of $(x^{*},y^{*}),(x_{new}^{*},y_{new}^{*})$.

Finally, for the new $(x^{*},y^{*})$ note that $T_{u}$ may have
changed in size but this is okay by design since $(k_{u}-|T_{u}|)x_{u}^{t}$
changes exactly by $(\#\text{new elements in }T_{u})\cdot x_{u}^{t}$.
\end{proof}
\begin{lem}
\label{lem:feas2}The final solution $(\bar{x},\bar{y})$ is feasible
with $\bar{x}$ integral.
\end{lem}
\begin{proof}
The fact that $\bar{x}$ is integral simply follows from the description
of the algorithm. The capacity constraints are satisfied by Lemma
\ref{lem:capFeas2} and the fact that any $u\in U$ would be rounded
up to $\bar{x}_{u}\geq1$ eventually (Lemma \ref{lem:stayabove2})
and therefore can pay for covering the tight edges $T_{u}$.

The other constraints $0\leq x_{v}\leq m_{v},y\geq0$ are guaranteed
by step 5 of the algorithm where we stop at $(x^{t},y^{t})$ before
they can be violated.
\end{proof}
Finally we prove the approximation guarantee, which again mostly follows
from examining the structure of an extreme point solution.
\begin{lem}
\label{lem:extpt2}When the algorithm terminates, $|W|\leq|U^{=}|$
where $U^{=}:=\{u\in U:x_{u}^{*}=m_{u}\}$.
\end{lem}
\begin{proof}
This is very similar to the previous proof. Note that the number of
variables is $|W|+|U|$. On the other hand, $0\leq x_{w}\leq1/f$
and $x_{u}\geq1/f$ cannot be tight. The number of tight constrains
$x_{u}\leq m_{u}$ is $|U^{=}|$ while there can be at most $|U|$
tight first covering constraints. So $|W|\leq|U^{=}|$.
\end{proof}
\begin{thm}
Our algorithm is a $f$-approximation.
\end{thm}
\begin{proof}
This is very similar to Theorem \ref{thm:main}. The cost of rounding
up intermediate $x_{v}^{t}=1/f$ incurs a factor of $f$. Since $(x^{*},y^{*})$
is feasible by design and Lemma \ref{lem:capFeas2}, it remains to
show that the last step incurs a factor of at most $f$. The same
inequality in the proof of Theorem \ref{thm:main} works.
\end{proof}

\section{Even Faster Implementation for Graphs ($f=2$)}

For the case of graphs $f=2$ we may in fact drop the variables $x_{u}$
from \textbf{LP3} altogether. This is exactly the same covering LP
used in CGW's two-stage rounding algorithm. In this section we show
that by performing iterative rounding while moving from old to new
optima, their algorithm actually achieves a 2-approximation for graphs.

In our opinion, this approach nicely explains why the randomized rounding
scheme in \cite{GandhiHKKS06} would work. The algorithm of \cite{GandhiHKKS06}
gives a 2-approximation for VCHC on simple graphs by essentially rounding
\textbf{LP2} via a simple randomized rounding and performing some
patching work after. However, a drawback of their work is that the
analysis involves many pages of calculations and gives little insights
into why the algorithm would give a 2-approximation.

In a way, their argument is a probabilistic proof that a 2-approximate
solution exists for \textbf{LP4}, and our one-page analysis can be
viewed as a simple deterministic proof of that claim.

\paragraph{LP relaxation}

Instead of \textbf{LP3}, we use the following simpler \textbf{LP4}:

\begin{subequations}\label{LP2-1} 
\begin{alignat}{2}
\text{min } & \qquad\qquad\sum_{w\in W}x_{w}\nonumber \\
\text{s.t. } & \sum_{w\in W}M(u,w)x_{w}\geq\sum_{w\in W}M(u,w)x_{w}^{*} & \quad & \forall u\in U\label{eq:2.1-1}\\
 & \qquad\qquad0\leq x_{w}\leq1 & \quad & \forall w\in W\label{eq:2.2-1}
\end{alignat}

\end{subequations}where 
\[
M(u,w)=\sum_{e=uw\in E\backslash T}\frac{y^{*}(e,w)}{x_{w}^{*}}.
\]

We first give an overview of the algorithm. The algorithm is mostly
a specialization of the previous one to $f=2$ but the constraint
$x_{u}\geq1/2$ is now redundant as an edge cannot have both endpoints
in $W$ for graphs. One crucial difference in the algorithm is that
we are rounding up $y^{*}(e,u)$ whenever $y^{*}(e,u)\cdot(\lceil x_{u}^{*}\rceil/x_{u}^{*})\geq1$
(rather than just when $x_{u}^{*}=y^{*}(e,u)$). This is still okay
as the final capacity is $k_{u}\lceil x_{u}^{*}\rceil$ so we can
blow up the coverage by a factor of $\lceil x_{u}^{*}\rceil/x_{u}^{*}$.

As for the analysis, to bound the cost of the extreme point solution
it is not enough to use only a cardinality inequality like $|W|\leq|U^{=}|$.
Instead we need a finer \textit{matching} structure between $U$ and
$W$ that was also used by CGW.

Our algorithm is as follows. Unlike before, \textit{we never redefine
or update $U$ and $W$.}

\begin{center}
\noindent\fbox{\begin{minipage}[t]{1\columnwidth - 2\fboxsep - 2\fboxrule}%
\begin{center}
\textbf{Iterative Rounding Algorithm using Simpler Covering LPs for
Graphs}
\par\end{center}
\begin{enumerate}
\item Solve \textbf{LP1} for $x^{*},y^{*}$. Let $T=\emptyset$.
\item For any $u\in U$ and $e\in\delta(u)\backslash T$ satisfying $y^{*}(e,u)\cdot(\lceil x_{u}^{*}\rceil/x_{u}^{*})\geq1$,
set $T\longleftarrow T\cup\{e\}$.
\item Solve updated \textbf{LP4} for an extreme point optimum $x_{new}^{*}$.
\item Consider $x^{t}=(1-t)x^{*}+tx_{new}^{*}$, where $t$ \textbf{continuously}
increases from 0 to 1. Let
\[
y^{t}(e,w)=\frac{y^{*}(e,w)}{x_{w}^{*}}x_{w}^{t}\text{, }y^{t}(e,u)=1-y^{t}(e,w)\quad\forall e=uw\in(U\times W)\backslash T
\]
\item Stop whenever $y^{t}(e,u)\cdot(\lceil x_{u}^{*}\rceil/x_{u}^{*})=1$
for some $u\in U$ and $e\in(U\times W)\backslash T$, set $T\longleftarrow T\cup\{e\}$.
\item Update $M(u,w)$ with $(x^{*},y^{*})\longleftarrow(x^{t},y^{t})$.
Repeat steps 3-6.
\item If $t=1$ (i.e. no pair $(e,u)$ satisfies the condition in step 5
for any intermediate $t$), output $\bar{x}_{v}=\lceil x_{v}^{*}\rceil$
with a corresponding $y$ defined by 
\[
\bar{y}(e,u)=y^{*}(e,u)\text{, }\bar{y}(e,v)=y^{*}(e,v)\quad\forall e\notin T
\]
\end{enumerate}
\[
\bar{y}(e,u)=1,\bar{y}(e,v)=0\quad\forall e=uw\in T,u\in U,w\in W
\]
\end{minipage}}
\par\end{center}

One difference is that we are rounding up $y^{*}(e,u)$ whenever $y^{*}(e,u)\cdot(\lceil x_{u}^{*}\rceil/x_{u}^{*})\geq1$
(rather than just when $x_{u}^{*}=y^{*}(e,u)$). This is still okay
as the final capacity is $k_{u}\lceil x_{u}^{*}\rceil$ so we can
blow up the coverage by a factor of $\lceil x_{u}^{*}\rceil/x_{u}^{*}$.
Curiously this modification is needed to handle the case $1<x_{u}^{*}\leq2$
when bounding the approximation ratio.

The feasibility of the algorithm is largely the same as before so
we skip the proof (in fact, it is more like an easier special case).
To bound the cost of the extreme point solution, it is not enough
to use only a cardinality inequality like $|W|\leq|U^{=}|$ (one fact,
now one does not have this but $|W|\leq|U|$). Instead we need a finer
\textit{matching} structure between $U$ and $W$ as given in Lemma
\ref{lem:extpt3}.
\begin{lem}
$(\bar{x},\bar{y})$ is a feasible solution to VCHC.
\end{lem}
\begin{proof}
Similar to Lemmas \ref{lem:capFeas2} and \ref{lem:feas2}.
\end{proof}
\begin{lem}[Similar to Theorem 3.2 of \cite{cheung2014improved}]
\label{lem:extpt3}For any extreme point solution $x^{*}$ to LP4,
let $W_{f}=\{w\in W:0<x_{w}^{*}<1\}$. Then there exists a matching
between $W_{f}$ and $U$ that fully matches $W_{f}$.
\end{lem}
\begin{proof}
Fractional $x_{w}^{*}$ can only arise from setting $|W_{f}|$ covering
constraints $\sum_{w\in W}M(u,w)x_{w}\geq\sum_{w\in W}M(u,w)x_{w}^{*}$
tight. Consider the system of equations $Ax=b$ obtained from setting
these $|W_{f}|$ constraints tight. Note that $A$ is $|W_{f}|\times|W_{f}|$
with the rows and columns indexed by $|U|$ and $|W_{f}|$ respectively.
$A$ is invertible so 
\[
\det A=\sum_{\text{permutation }\pi}\left(\pm\prod_{i}a_{i,\pi(i)}\right)\neq0
\]
 which implies that there is a permutation with all $a_{i,\pi(i)}\neq0$.
Such a permutation gives our desired matching since nonzero entries
of $A$ corresponds to $M(u,w)\neq0$ in which case $uw$ is an edge.
\end{proof}
\begin{thm}
$(\bar{x},\bar{y})$ is a 2-approximation.
\end{thm}
\begin{proof}
The cost of $U$ and $x_{v}=0$ clearly incurs a factor 2 blowup only.
We need to account for the cost of selecting $W_{f}$. Now that we
don't have $|W_{f}|\leq|U^{=}|$, we cannot naively charge $|W_{f}|$
to $U^{=}$. However, we have a matching between $W_{f}$ and $U$
and it suffices to show that for any edge $e=uw$ in the matching,
\[
\lceil x_{u}^{*}\rceil+1=\lceil x_{u}^{*}\rceil+\lceil x_{w}^{*}\rceil\leq2(x_{u}^{*}+x_{w}^{*}).
\]
Note that we have $1=y^{*}(e,u)+y^{*}(e,w)\leq x_{u}^{*}+x_{w}^{*}$
and $y^{*}(e,u)\cdot(\lceil x_{u}^{*}\rceil/x_{u}^{*})\leq1$ (otherwise
step 5 would have applied). We have three cases:

\begin{itemize}
\item $\lceil x_{u}^{*}\rceil=1$. Then we are done.
\item $\lceil x_{u}^{*}\rceil\geq3$. Then $\lceil x_{u}^{*}\rceil+1\leq2x_{u}^{*}$.
\item $x_{u}^{*}+x_{w}^{*}\geq1.5$ and $x_{u}^{*}\leq2$. Then $\lceil x_{u}^{*}\rceil+1\leq3\leq2(x_{u}^{*}+x_{w}^{*})$.
\item $x_{u}^{*}+x_{w}^{*}\leq1.5$ and $1<x_{u}^{*}\leq2$, which cannot
happen since $1=y^{*}(e,u)+y^{*}(e,w)\leq x_{u}^{*}/2+x_{w}^{*}=x_{u}^{*}+x_{w}^{*}-x_{u}^{*}/2<1.5-1/2=1$.
Contradiction.
\end{itemize}
\end{proof}
In hindsight, CGW came close but failed to obtain a 2-approximation
because one does not have $x_{new,u}^{*}+x_{new,w}^{*}\geq1$.

\section{Open Problem}

In this paper we have settled the approximability of VCHC. While it
is UGC-hard to do better \cite{KhotR08}, the current best NP-hardness
inapproximability stands at $f-1$ \cite{DinurGKR05} and 1.36 for
graphs \cite{dinur2005hardness}. Is there any hope of proving that
it is NP-hard to beat $f$? In principle it should be easier than
doing it for vertex cover since VCHC is more general.

\section*{Acknowledgement}

We thank Wang Chi Cheung, Lap Chi Lau and Christos Papadimitriou for
several helpful discussions. We also thank the anonymous reviewers
for many helpful comments that improved the presentation of this paper.

\bibliographystyle{plain}
\bibliography{vchcref}
 
\end{document}